\documentclass[letterpaper,twocolumn,twoside]{IEEEtran} 

\usepackage{amsmath, amssymb,bm}
\usepackage{enumerate,mathtools}
\usepackage{amsthm}
\usepackage{microtype}
\usepackage{cite}

\usepackage{tikz}
\usetikzlibrary{shapes.geometric, arrows}

\newtheorem{theorem}{Theorem}

\newtheorem{lemma}[theorem]{Lemma}
\newtheorem{prop}[theorem]{Proposition}

\theoremstyle{definition}

\newcommand{\Beta}{\mathrm{B}}

\newcommand{\cV}{\mathcal{V}}

\global\long\def\dd{\mathrm{d}}

\newcommand{\one}{\boldsymbol{1}}

\newcommand{\ex}[1]{\ensuremath{\mathbb{E}\left[ #1\right]}}
\newcommand{\pr}[1]{\ensuremath{\mathbb{P}\left[ #1\right]}}

\DeclareMathOperator{\var}{\sf Var}

\newcommand{\DKL}[2]{\ensuremath{D\left( #1 \, \middle \| \, #2 \right)}}

\DeclareMathOperator{\vol}{Vol}
\DeclareMathOperator{\cone}{cone}

\newcommand{\reals}{\mathbb{R}}

\newcommand{\eps}{\epsilon}
\newcommand{\normal}{\mathcal{N}}
\newcommand{\gmid}{\! \mid \!}

\newcommand{\set}{S}

\let\originalleft\left
\let\originalright\right
\renewcommand{\left}{\mathopen{}\mathclose\bgroup\originalleft}
\renewcommand{\right}{\aftergroup\egroup\originalright}

  \newif\iflongpaper
 \longpapertrue

\title{Two-Moment Inequalities for R\'enyi Entropy\\ and Mutual Information}

\author{
Galen Reeves
\thanks{The work of G.\ Reeves was supported in part by funding from the Laboratory for Analytic Sciences (LAS).  Any opinions, findings, conclusions, and recommendations expressed in this material are those of the author and do not necessarily reflect the views of the sponsors. }
\thanks{G. Reeves is with the Department of Electrical and Computer Engineering and the Department of Statistical Science, Duke University, Durham (e-mail: galen.reeves@duke.edu)}
}

\begin{document}


\maketitle

\begin{abstract}

This paper explores some applications of a two-moment inequality for the integral of the $r$-th power of a function, where $0  < r< 1$. The first contribution is an upper bound on the R\'enyi entropy of a random vector in terms of the two different moments. When one of the moments is the zeroth moment, these bounds recover previous results based on maximum entropy distributions under a single moment constraint. More generally, evaluation of the bound with two carefully chosen nonzero moments can lead to significant improvements with a modest increase in complexity. The second contribution is a method for upper bounding mutual information in terms of certain integrals with respect to the variance of the conditional density. The bounds have a number of useful properties arising from the connection with variance decompositions.
\end{abstract}


\section{Introduction}

Measures of entropy and information play a central role in applications throughout information theory,  statistics, computer science, and statistical physics. In many cases, there is interest in understanding maximal properties of these measures over a  given family of distributions. One example is given by the principle of maximum entropy, which originated in statistical mechanics and was introduced in broader context by Jaynes~\cite{jaynes:1982}. 

Entropy-moment inequalities can be used to describe properties of distributions characterized by moment constraints. Perhaps  the most well known entropy-moment inequality follows from the fact that the Gaussian distribution maximizes differential entropy over all distributions with the same variance \cite[Theorem~8.6.5]{cover:2006}. This inequality leads to remarkably simple proofs for fundamental results in information theory and estimation theory. 

A variety of entropy-moment inequalities have also been studied in the context of R\'enyi entropy \cite{costa:2002,lutwak:2004,lutwak:2007,johnson:2007, lutwak:2013}, which is a generalization of Shannon entropy. Recent work has focused on the extremal distributions for the closely related R\'enyi divergence~\cite{erven:2014,kumar:2015a,kumar:2015b, sason:2016,bobkov:2016}. 

Another line of work focuses on relationships between measures of dissimilarity between probability distributions provided by the family of $f$-divergences \cite{nielsen:2014, sason:2016a}, which includes as special cases, the total variation distance, relative entropy (or Kullback-Leibler divergence), R\'enyi divergence, and chi-square divergence. One application of these results is to provide bounds for mutual information in terms of divergence measures that dominate relative entropy, such as the chi-square divergence; see e.g.\ \cite{nielsen:2014, shao-lun:2015}.

\subsection{Overview of results}


The starting point of our analysis (Proposition~\ref{prop:two_moment}) is an inequality for the integral of the $r$-th power of a  function. Specifically, for any numbers $p,q,r$ with 
\[
0 < r< 1 \quad \text{and} \quad  p < \frac{ 1-r}{r} < q,
\]
the following inequality holds:
\[
\left( \int \! f^r(x) \, \dd x \! \right)^{\!\frac{1}{r}}  \!\!\!
 \le C\,  \left(\int \! |x|^p f(x) \, \dd x \! \right)^{\!\lambda} \! \left( \int\! |x|^q f(x) \, \dd x \! \right)^{\!1-\lambda}\!\!,
\]
for all  non-negative functions $f: \reals_+ \to \reals_+$ where $C$ and $0 < \lambda < 1$ are given explicitly in terms of the tuple $(p,q,r)$. An extension to functions defined on an arbitrary subset of $\reals^n$ is also provided (Proposition~\ref{prop:two_moment_n}).

The remainder of the paper shows how this inequality can be used to provide bounds on information measures such as R\'enyi entropy and mutual information.  Some useful properties of the bounds include:
\begin{itemize}

\item  \textit{Simplicity:} Beyond the existence of a density, these bounds do not require further regularity conditions such as boundedness or sub-exponential tails. As a consequence, these bounds can be applied under relatively mild technical assumptions.   

\item \textit{Tightness:} For some applications, the bounds can provide an accurate characterization of the underlying information measures. For example, a special case of Proposition~\ref{prop:MI_two_moment_inq} in this paper played a key role  in the author's recent work \cite{reeves:2016,reeves:2016a,reeves:2016b}, where it was used to bound the relative entropy between low-dimensional projections of a random vector and a Gaussian approximation.

\item \textit{Geometric Interpretation:} Our bounds on the mutual information between random variables $X$ and $Y$ can be expressed in terms of the variance of the conditional density of $Y$ given $X$. Specifically, the bounds depend on integrals of the form:   
\[
 \int   \|y\|^{s} \var(f_{Y|X}(y|X)) \, \dd y.
\]
For $s=0$, this integral is the expected squared $L^2$ distance between the conditional density $f_{Y|X}$ and the marginal density $f_Y$. 
\end{itemize}

The paper is organized as follows: Section~\ref{sec:moment_inequalites} provides integral inequalities for nonnegative functions; Section~\ref{eq:entropy_bounds} gives bounds on R\'enyi entropy of orders less than one; and Section~\ref{sec:MI_bounds} provides bounds on mutual information. 

\section{Moment Inequalities} \label{sec:moment_inequalites}

Throughout this section, we assume that $f$ is a real-valued Lebesgue measurable function defined on a measurable subset $S$ of $\reals^n$. For any positive number $p$, the function $\|\cdot\|_p$ is defined according to 
\[
\|f\|_p = \left( \int_S \left| f(x) \right|^p \, \dd x \right)^\frac{1}{p}.
\]
Recall that for $0 <p < 1$, the function $\|\cdot\|_p$ is not a norm because it does not satisfy the triangle inequality. The $s$-th moment of $f$ is defined according to
\[
\mu_s(f)  = \int_{S}  \|x\|^s \, f(x) \, \dd x,
\]
where $\|\cdot\|$ denotes the standard Euclidean norm on vectors.

\subsection{Multiple Moments} 
Consider the following  optimization problem:
\begin{alignat*}{3}
 \text{maximize}  \quad & \|f\|_r \\
 \text{subject to} \quad & f(x) \ge  0 &\quad &   \text{for all $x \in S$}   \label{eq:optimization_prob}\\
 &   \mu_{s_i}(f)  \le m_i & &  \text{for $1\le i \le k$}.
\end{alignat*}
For $r \in (0,1)$ this is a convex optimization problem because $\|\cdot\|_r^r$ is concave and the moment constraints are linear. By standard theory in convex optimization\iflongpaper \ (see e.g., \cite{rockafellar:1970})\fi, it can be shown that if the problem is feasible and the maximum is finite, then the maximizer has the form
\[
f^*(x) = \bigg(  \sum_{i = 1}^k \nu^*_i \,  \|x\|^{s_i} \bigg)^\frac{1}{r-1} , \quad \text{for all $x \in S$}.  
\]
The parameters $\nu^*_1, \cdots, \nu^*_k$ are nonnegative and the $i$-th moment constraint holds with equality for all  $i$ such that $\nu^*_i $ is strictly positive, that is $\nu^*_i > 0 \implies \mu_{s_i}(f^*) = m_i$. Consequently, the maximum can be expressed in terms of a linear combination of the moments: 
\[
\|f^*\|^r_r = \|(f^*)^r\|_1 = \|f^* (f^*)^{r-1} \|_1   =  \sum_{i=1}^k \nu_i^* m_i.
\]

For the purposes of this paper, is it is useful to consider a relative inequality in terms of the moments of the function itself. Given a number $0 < r < 1$ and vectors $s \in \reals^k$ and $\nu \in \reals_+^k$ the function $c_r(\nu,s)$ is defined according to 
\[
c_r(\nu,s)  = \left(  \int_0^\infty \bigg( \sum_{i=1}^k \nu_i \, x^{s_i}  \bigg)^{-\frac{r}{1-r} }  \dd x \right)^\frac{1-r}{r},
\]
if the integral exists. Otherwise, $c_r(\nu,s)$ is defined to be positive infinity. It can be verified that $c_r(\nu,s)$ is finite provided that there exists $i,j$ such that $\nu_i$ and $\nu_j$ are strictly positive and $s_i < (1-r)/r < s_j$. 

The following result can be viewed as a consequence of the constrained optimization problem described above. We provide a different and very simple proof that depends only  on H\"older's inequality. 

\begin{prop}\label{prop:k_moments} 
Let $f$ be a  nonnegative Lebesgue measurable function defined on the positive reals $\reals_+$. For any number $0 < r< 1$ and vectors $s \in \reals^k$ and $\nu \in \reals_+^k$, we have
\[
\| f\|_r \le c_r(\nu,s) \sum_{i=1}^k \nu_i \, \mu_{s_i} (f) .
\]
\end{prop}
\begin{proof} Let $g(x ) =\sum_{i = 1}^k \nu_i \,  x^{s_i}$.  Then, we have
\begin{align*}
\|f\|_r^r   & =\|  g^{-r}   ( f g)^r \|_1\\
& \le \| g^{-r} \|_\frac{1}{1-r} \| (gf)^r \|_\frac{1}{r} \\
& =  \| g^\frac{-r}{1-r}  \|_1^{1-r}  \| gf \|^r_1 \\
& =    \bigg( c_r(\nu,s)\,  \sum_{i=1}^k \nu_i \, \mu_{s_i} (f)  \bigg)^r,
\end{align*}
where second step follows from  H\"older's inequality with conjugate exponents  $1/(1-r)$ and $1/r$.
\end{proof}

\subsection{Two Moments} 

The next result follows from  Proposition~\ref{prop:k_moments} for the case of two moments.

\begin{prop}\label{prop:two_moment}
Let $f$ be a  nonnegative Lebesgue measureable function defined on the positive reals $\reals_+$.  For any numbers $p,q,r$ with $0 < r <1$ and $p < 1/r - 1 < q$, we have
\[
\|f\|_r \le \left[ \psi_{r}(p,q) \right]^\frac{1-r}{r}  \,  [\mu_p(f)]^{\lambda}   [\mu_q(f)]^{1-\lambda},
\]
where  $\lambda = (q + 1 -  1/r  )/(q-p)$ and 
\begin{align}
\psi_r(p,q) = \frac{1}{(q-p)} \widetilde{\Beta}\left( \frac{r \lambda }{ 1- r}   ,  \frac{r(1 -  \lambda) }{ 1- r} \right), \label{eq:psi_r} 
\end{align}
where $\widetilde{B}(a,b) = \Beta(a,b) (a+b)^{a+b} a^{-a} b^{-b}$ and $\Beta(a,b)$ is the Beta function. 
\end{prop}
\iflongpaper
\begin{proof}
Letting $s = (p,q)$ and $\nu = (\gamma^{1-\lambda} , \gamma^{-\lambda})$ with $\lambda >0$, we have
\[
[c_r(\nu, s)]^\frac{r}{1-r}  =  \int_0^\infty \left( \gamma^{1-\lambda} \,  x^p  + \gamma^{-\lambda} \,  x^q\right)^{-\frac{r}{1-r} } \, \dd x.
\]
Making the change of variable $x \mapsto ( \gamma  u )^\frac{1}{q-p}$ leads to 
\[
[c_r(\nu, s)]^\frac{r}{1-r}  = \frac{ 1  }{(q-p)}  \int_0^\infty \frac{ u^{b- 1}}{ (1+u)^{a+b} } \, \dd u = \frac{ \Beta \left(a   ,  b \right) }{(q-p)} ,
\]
where $a = \frac{r}{1-r} \lambda$ and $b = \frac{r}{1-r} (1-\lambda)$ and the second step follows from the integral representation of the Beta function \cite[Eq.~(1.1.19)]{andrews:1999}. Therefore, by Proposition~\ref{prop:k_moments}, the inequality
\[
\|f\|_r  \le \left( \frac{  \Beta \left(a   ,  b \right) }{q-p}  \right)^\frac{1-r}{r}  \left( \gamma^{1-\lambda}  \mu_p(f) + \gamma^{-\lambda} \, \mu_q(f)    \right),
\]
holds for all $\gamma > 0$.  Evaluating this inequality  with
\[
\gamma  = \frac{ \lambda \, \mu_q(f) }{ (1-\lambda) \mu_p(f)},
\]
leads to the stated result. 
\end{proof}
\fi

The special case $r = 1/2$ admits the simplified expression
\begin{equation}
\psi_{1/2}(p,q)  = \frac{ \pi  \lambda^{-\lambda} (1-\lambda)^{-(1-\lambda)}}{ (q-p) \sin( \pi \lambda)}, \label{eq:psi_one_half}
\end{equation}
where we have used Euler's reflection formula for the Beta function \cite[Theorem~1.2.1]{andrews:1999}.

Next, we consider an extension of Proposition~\ref{prop:two_moment} for functions defined on $\reals^n$. Given any measurable subset $S$ of $\reals^n$ we define
\begin{equation}
\omega(S)  =  \vol( B^n   \cap \mathrm{cone}( \set)) \label{eq:oemgaS},
\end{equation} 
where $B^n = \{ u \in \reals^n  : \|u\| \le 1\}$ is the $n$-dimensional Euclidean ball of radius one and 
\[
 \mathrm{cone}( \set) = \{ x \in \reals^n \, : \, \text{$t x \in \set $ for some $t > 0$}\}.
\]
The function $\omega(S)$ is proportional to the surface measure of the projection of $\set$ on the Euclidean sphere and satisfies 
\begin{equation}
\omega(S) \le \omega(\reals^n) =  \frac{\pi^\frac{n}{2} }{\Gamma(\frac{n}{2} +1)}, \label{eq:omega_UB}
\end{equation}
for all $S \subseteq \reals^n$. Note that $\omega(\reals_+) = 1$ and $\omega(\reals) = 2$. 

\begin{prop}\label{prop:two_moment_n} Let $f$ be a  nonnegative Lebesgue measurable function defined on a subset $\set$ of $\reals^n$. 
For any numbers $p,q,r$ with  $0 < r< 1$ and $p < 1/r -1 < q$,  we have
\[
 \|f\|_r \le  \left[ \omega(\set)  \,  \psi_r(p,q)  \right]^{\frac{1-r}{r}}\,  [\mu_{np}(f)]^{\lambda}   [\mu_{nq}(f)]^{1-\lambda}  ,
\]
where $\lambda = (q + 1 -  1/r)/(q-p)$ and $\psi_r(p,q)$ is given by \eqref{eq:psi_r}. 
\end{prop}
\iflongpaper
\begin{proof}
Let $f$ be extended to $\reals^n$ using the rule $f(x)  = 0$ for all $x$ outside of $S$ and let $g: \reals_+ \to \reals_+$ be defined according to 
\[
g(y)  = \frac{1}{n}   \int_{\mathbb{S}^{n-1}} f( y^\frac{1}{n}  u)  \, \dd \sigma(u),
\]
where $\mathbb{S}^{n-1} = \{u \in \reals^n \, : \, \|u\| =1\}$ is the Euclidean sphere of radius one and $\sigma(u)$ is the surface measures of the sphere. We will show that
\begin{align}
\|f\|_r & \le \left( \omega(S) \right)^\frac{1-r}{r} \|g\|_r \label{eq:two_moment_n_a}\\
\mu_{ns}(f) & = \mu_s(g). \label{eq:two_moment_n_b}
\end{align}
Then, the stated inequality then follows from applying Proposition~\ref{prop:two_moment} to the function $g$. 

In order to prove \eqref{eq:two_moment_n_a}, we begin with a transformation into polar coordinates:
\begin{equation}
\|f\|^r_r = \int_0^\infty  \int_{\mathbb{S}^{n-1}} \left| f(t  u)\right|^r  t^{n-1}  \, \dd \sigma(u) \, \dd t. \label{eq:two_moment_n_c} 
\end{equation}
Letting $\one_{\cone(S)}(x)$ denote the indicator function of the set $\cone(S)$, the integral over the sphere can be  bounded using:
\begin{align}
\MoveEqLeft[1] \int_{\mathbb{S}^{n-1}} \left| f(t  u)\right|^r  \, \dd \sigma(u) \notag \\
& = \int_{\mathbb{S}^{n-1}}  \one_{\cone(S)}(u) \, \left| f(t  u)\right|^r  \, \dd \sigma(u) \notag \\
& \overset{(a)}{\le} \left( \int_{\mathbb{S}^{n-1}} \! \one_{\cone(S)}(u) \,  \dd \sigma(u) \! \right)^{\! 1-r} \left(  \int_{\mathbb{S}^{n-1}}\left|  f(t  u) \right|  \, \dd \sigma(u) \! \right)^r \notag \\
& \overset{(b)}{=} n\,  \left( \omega(S) \right)^{1-r} g^r(t^n).  \label{eq:two_moment_n_d} 
\end{align}
where: (a) follows from H\"older's inequality with conjugate exponents $\frac{1}{1-r}$ and $\frac{1}{r}$; and (b) follows from the definition of $g$ and the fact that 
\begin{align*}
\omega(S) & = \int_0^1 \int_{\mathbb{S}^{n-1}} \! \one_{\cone(S)}(u) \,t^{n-1} \,   \dd \sigma(u)\, \dd t\\
& = \frac{1}{n} \int_{\mathbb{S}^{n-1}} \! \one_{\cone(S)}(u) \,  \dd \sigma(u).
\end{align*}
Plugging \eqref{eq:two_moment_n_d} back into \eqref{eq:two_moment_n_c} and then making the change of variables $t \to y^\frac{1}{n}$ yields 
\[
\|f\|^r_r  \le n\,  \left( \omega(S)\right)^{1-r}  \int_0^\infty g^r(t^n) t^{n-1} \, \dd t  = \left( \omega(S)\right)^{1-r} \|g\|_r^r.
\]

The proof of  \eqref{eq:two_moment_n_b} follows along similar lines. We have 
\begin{align*}
\mu_{ns}(f) 
& \overset{(a)}{=} \int_0^\infty \int_{\mathbb{S}^{n-1}} t^{ns}  f(t u )\,  t^{n-1}\, \dd \sigma(u)\,  \dd t  \\
& \overset{(b)}{=} \frac{1}{n} \int_0^\infty \int_{\mathbb{S}^{n-1}} y^{s}  f(y^\frac{1}{n}  u )  \, \dd \sigma(u)\, \dd y  \\
& = \mu_{s}(g) 
\end{align*}
where (a) follows from a transformation into polar coordinates and (b) follows form the change of variable $t \mapsto  y^\frac{1}{n}$.
\end{proof}
\fi

\section{R\'enyi Entropy Bounds}\label{eq:entropy_bounds} 

Let $X$ be a random vector that has a density $f(x)$ with respect to Lebesgue measure on  $\reals^n$. The differential R\'enyi entropy of order $r \in (0,1) \cup (1,\infty)$ is defined according to \cite{cover:2006}:
\[
h_r(X)  = \frac{1}{1-r} \log\left(  \int_{\reals^n} f^r(x) \, \dd x\right).
\]
The R\'enyi entropy is continuous and non-increasing in $r$. If the support set $S = \{x \in \reals^n  : f(x) > 0\}$ has finite measure then the limit as $r$ converges to zero is given by $h_0(X) = \log \vol(S)$. \iflongpaper If the support does not have finite measure then $h_r(X)$ increases to infinity as $r$ decreases to zero. \fi The case $r=1$ is given by the Shannon differential entropy:
 \[
h_1(X) = -  \int_S f(x) \log f(x) \, \dd x.
\]

Given a random variable $X$ that is not identically zero and numbers $p,q,r$ with $0<r<1$ and $p < 1/r - 1 < q$, we define the function
\[
L_r(X; p,q) =  \frac{r \lambda }{1-r}  \log \ex{ |X|^p } +\frac{r (1 - \lambda) }{1-r}   \log \ex{ |X|^q },
\]
where $\lambda = (q + 1 - 1/r )/(q-p)$.

The next result, which follows directly from Proposition~\ref{prop:two_moment_n}, provides an upper bound on the  R\'enyi entropy.

\begin{prop}\label{prop:entropy_bound}
Let $X$ be a random vector with a  density  on $\reals^n$. For any numbers $p,q,r$ with $0 <r<1$ and $p < 1/r -1 < q$, the R\'{e}nyi entropy satisfies
\begin{equation}
h_r(X)  \le \log \omega(S)  + \log \psi_r(p,q)  + L_{r}(\|X\|^n; p,q), \label{eq:entropy_bound}
\end{equation}
where $\omega(S)$ is defined in~\eqref{eq:oemgaS}  and $\psi_r(p,q)$ is defined in \eqref{eq:psi_r}.
\end{prop}
\iflongpaper
\begin{proof}
This result follows immediately from Proposition~\ref{prop:two_moment_n} and the definition of R\'enyi entropy. 
\end{proof}
\fi

The relationship between Proposition~\ref{prop:entropy_bound} and previous results depends on whether the moment $p$ is equal to zero:
\begin{itemize}
\item \textit{One-moment inequalities:} If $p=0$ then there exists a distribution such that \eqref{eq:entropy_bound} holds with equality. This is because the zero-moment constraint ensures that the function that maximizes the R\'enyi entropy integrates to one. In this case, Proposition~\ref{prop:entropy_bound} is equivalent to previous results that focused on distributions that maximize R\'enyi entropy subject to a single moment constraint \cite{costa:2002,lutwak:2004,lutwak:2007}. With some abuse of terminology we refer to these bounds as one-moment inequalities\footnote{A more accurate name would be two-moment inequalities under the constraint that one of the moments is the zeroth moment.}. 

\item \textit{Two-moment inequalities:} If $p \ne 0$ then the right-hand side of \eqref{eq:entropy_bound} corresponds to the R\'{e}nyi entropy of a non-negative function that might not integrate to one. Nevertheless, the expression provides an upper bound on the R\'enyi entropy for any density with the same moments. We refer to the bounds obtained using a general pair $(p,q)$ as two-moment inequalities. 
\end{itemize}

The contribution of  two-moment inequalities is that they lead to tighter bounds. To quantify the tightness, we define $\Delta_r(X; p,q)$ to be the gap between the right-hand side and left-hand side of \eqref{eq:entropy_bound} corresponding to the pair $(p,q)$, that is
\begin{align*}
\Delta_r(X; p,q)  & =  \log \omega(S)  + \log \psi_r(p,q) \\
& \quad  + L_{r}(\|X\|^n; p,q) - h_r(X). 
\end{align*}
The gaps corresponding to the optimal two-moment and one-moment inequalities are defined according to:
\begin{align*}
\Delta_r(X) & = \inf_{p ,q} \Delta_r(p,q) \\
\widetilde{\Delta}_r(X) & = \inf_{q} \Delta_r(0,q).
\end{align*}

\subsection{Some consequences of these bounds} 

By Lyapunov's inequality, the mapping $s \mapsto \frac{1}{s} \log\ex{ |X|^s}$ is nondecreasing  on $[0,\infty)$ and thus
\begin{equation}
L_r(X; p,q)  \le  L_r(X; 0,q) =  \frac{1}{q}  \log \ex{ |X|^q}, \quad p \ge 0.  \label{eq:L_UB}
\end{equation}
In other words, the case $p = 0$ provides an upper bound on $L_r(X; p,q)$ for nonnegative $p$.  Alternatively, we also have the lower bound 
\begin{equation}
L_r(X; p,q)  \ge  \frac{r }{1-r}  \log \ex{ |X|^\frac{1-r}{r} },
\end{equation}
which follows from the convexity of  $\log \ex{ |X|^s}$.

A useful property of $L_r(X; p,q)$  is that it is additive with respect to the product of independent random variables. Specifically, if $X$ and $Y$ are independent, then
\begin{align}
L_r(X Y ; p,q) = L_r(X; p,q) + L_r(Y; p,q). \label{eq:L_additive}
\end{align}
One consequence is that 
multiplication by a bounded random variable cannot increase the R\'enyi entropy by an amount that exceeds the gap of the two-moment inequality with nonnegative moments.

\begin{prop}\label{prop:entropy_mult_bound}
Let $Y$ be a random vector on $\reals^n$ with finite R\'{e}nyi entropy of order $0 < r< 1$, and let $X$ be an independent random variable that satisfies $0 < X \le t$. Then, 
\[
h_r(XY) \le h_r(tY)  + \Delta_r(Y ; p,q),
\]
for all $0  < p < 1/r - 1 < q$. 
\end{prop}
\iflongpaper
\begin{proof}
Let $Z = XY$ and let $\set_Z$ and $\set_Y$ denote the support sets of $Z$ and $Y$, respectively. The assumption that $X$ is non-negative means that $\cone(\set_Z) = \cone(\set_Y)$. We have
\begin{align*}
h_r(Z) & \overset{(a)}{\le}  \log \omega(\set_Z)  + \log \psi_r(p,q)  + L_{r}(\|Z\|^n; p,q)\\
& \overset{(b)}{=} h_r(Y) + L_r(|X|^n; p;q)  +  \Delta_r(Y; p,q) \\
& \overset{(c)}{\le} h_r(Y)  + n \log t + \Delta_r(Y; p,q),
\end{align*}
where: (a) follows from Proposition~\ref{prop:entropy_bound}; (b) follows from \eqref{eq:L_additive} and the definition of $\Delta_r(Y ; p,q)$, and (c) follows from \eqref{eq:L_UB} and the assumption $|X|\le t$. Finally, recalling that $h_r(t Y) = h_r(Y) + n \log t$ completes the proof. 
\end{proof}
\fi


\subsection{Example with lognormal distribution} 
 If $W \sim \normal(\mu, \sigma^2)$ then the random variable $X = \exp(W)$ has a lognormal distribution with parameters $(\mu, \sigma^2)$. The R\'enyi entropy is given by
\[
h_r(X)  =  \mu   +  \frac{1}{2} \left( \frac{1-r}{r}  \right) \sigma^2 + \frac{1}{2} \log( 2 \pi r^\frac{1}{r-1}   \sigma^2 ),
\]
and the logarithm of the $s$-th moment is given by 
\[
\log \ex{|X|^s} =  \mu s\, +  \frac{1 }{2} \sigma^2 \, s^2. 
\]
With a bit of work, it can be shown that the gap of the optimal two-moment inequality does not depend on the parameters $(\mu,\sigma^2)$ and is given by 
%
\begin{align}
\Delta_r(X) & =  \log\left(  \widetilde{\Beta}\left( \frac{ r}{2(1- r)} ,  \frac{ r}{2(1- r)} \right)  \sqrt{  \frac{ r   }{4 (1-r)   }}  \right) \notag \\
& \quad  + \frac{1}{2}  -   \frac{1}{2} \log(2\pi  r^\frac{1}{r-1} ). \label{eq:Delta_lognormal}
\end{align}
The details of this derivation are given in Appendix~\ref{sec:lognormal_proof}. Meanwhile, the gap of the optimal one-moment inequality is given by
\begin{align}
\widetilde{\Delta}_r(X) & = \inf_{q  }  \bigg[ \log\left(  \widetilde{\Beta}\left( \frac{ r}{1- r} - \frac{1}{q}  , \frac{1}{q}  \right)   \frac{1}{q}  \right) + \frac{1}{2} q  \sigma^2  \bigg] \notag \\
& \quad - \frac{1}{2} \left( \frac{1-r}{r} \right) \sigma^2  -  \frac{1}{2} \log( 2 \pi r^\frac{1}{r-1} \sigma^2 ).
\end{align}

 The functions $\Delta_r(X)$ and $\widetilde{\Delta}_r(X) $ are illustrated in Figure~\ref{fig:lognormal} as a function of $r$ for various $\sigma^2$. The function $\Delta_r(X)$ is bounded uniformly with respect to $r$ and converges to zero  as $r$ increases to one. 
 The tightness of the two-moment inequality in this regime follows from the fact that the lognormal distribution maximizes Shannon entropy subject to a constraint on $\ex{ \log X}$.    By contrast, the function $\widetilde{\Delta}_r(X)$ varies with the parameter $\sigma^2$. For any fixed $r \in (0,1)$, it can be shown that $\widetilde{\Delta}_r(X)$ increases to infinity if $\sigma^2$ converges to zero or infinity.

\begin{figure}
\centering
\begin{tikzpicture}
    	\node[anchor=south west,inner sep=0] at (0,0) {\includegraphics{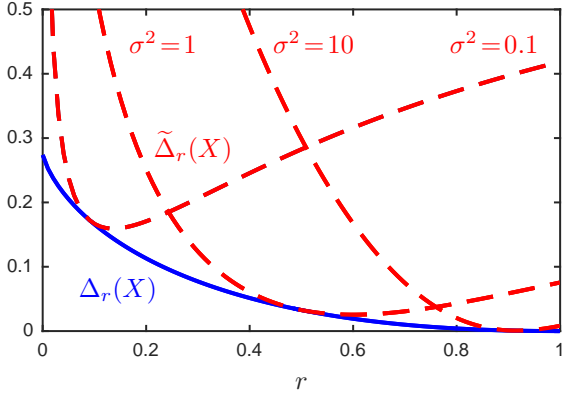}};
        \node[ align =center, text = {rgb:red,0;green,0;blue,1}] at (1.5,1.4) { $\Delta_r(X)$};
 	\node[ align =center, text = {rgb:red,1;green,0;blue,0}] at (2.5,3.3) { $\widetilde{\Delta}_r(X)$};
	\node[ align =center, text = {rgb:red,1;green,0;blue,0}] at (4.1,4.7) { $\sigma^2\! = \!10$};
	\node[ align =center, text = {rgb:red,1;green,0;blue,0}] at (2.1,4.7) { $\sigma^2 \!= \!1$};
	\node[ align =center, text = {rgb:red,1;green,0;blue,0}] at (6.5,4.7) { $\sigma^2\! =\! 0.1$};
\end{tikzpicture}
\vspace*{-.3 cm}
\caption{\label{fig:lognormal} Comparison of upper bounds on R\'enyi entropy for the lognormal distribution as a function of the order $r$ for various $\sigma^2$.} 
\label{default}
\end{figure}

\begin{figure}
\centering
\begin{tikzpicture}
    	\node[anchor=south west,inner sep=0] at (0,0) {\includegraphics{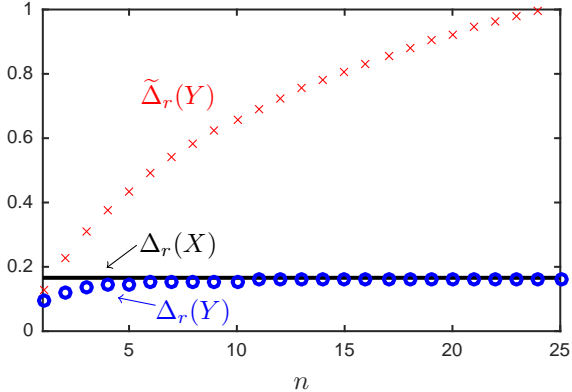}};
       \draw [<-,color = {rgb:red,0;green,0;blue,1}] (1.5, 1.3) -- (2,1.2);
        \node[ align =center, text = {rgb:red,0;green,0;blue,1}] at (2.5,1.1) { $\Delta_r(Y)$};
 	\node[ align =center, text = {rgb:red,1;green,0;blue,0}] at (2.3,4.0) { $\widetilde{\Delta}_r(Y)$};
	    \draw [->] (1.75, 2.0) -- (1.35,1.7);
	\node[ align =center] at (2.3,2.0) { $\Delta_r(X)$};
\end{tikzpicture}
\vspace*{-.3 cm}
\caption{\label{fig:gaussian} Comparison of upper bounds on R\'enyi entropy for the multivariate Gaussian distribution $\normal(0 , I_n)$ as a function of the dimension $n$ with $r = 0.1$. The solid black line is the gap of the optimal two-moment inequality for the lognormal distribution.}

\label{default}
\end{figure}

\subsection{Example with multivariate Gaussian distribution} 

Next, we consider the case where $Y\sim \normal(0, I_n)$ is an $n$-dimensional Gaussian vector with mean zero and identity covariance.  The R\'enyi entropy is given by
\[
h_r(Y)  = \frac{n}{2} \log( 2 \pi r^\frac{1}{r-1} ) ,
\]
and the $s$-th moment of the magnitude $\|X\|$ is given by 
\[
\ex{\|Y\|^s}   = \frac{ 2^\frac{s}{2} \Gamma( \frac{n + s}{2} )}{\Gamma( \frac{n}{2} )} .
\]

As the dimension $n$ increases, it can be shown that the gap of the optimal two-moment inequality converges to the gap for the lognormal distribution. The proof of the following result is given in Appendix~\ref{sec:Delta_Gaussian_proof}. 
\begin{prop}\label{prop:Delta_Gaussian_proof} 
If $Y \sim \normal(0, I_n)$ then, 
\[
\lim_{n \to \infty} \Delta_r(Y )  = \Delta_r(X),
\]
where $X$ has a lognormal distribution.
\end{prop}

The functions $\Delta_r(Y)$ and $\widetilde{\Delta}_r(Y)$ are illustrated in Figure~\ref{fig:gaussian}.  Both functions are increasing in the dimension $n$. However, while $\Delta_r(Y)$ converges to a finite limit, $\widetilde{\Delta}_r(Y)$ increases without bound. For any fixed integer $n$, it can be shown that both $\Delta_r(Y)$ and $\widetilde{\Delta}_r(Y)$ converge to zero as $r$ increases to one. This behavior follows from the fact that the Gaussian distribution is the maximum entropy distribution for Shannon entropy under a second moment constraint.

\subsection{Inequalities for differential entropy}

Proposition~\ref{prop:entropy_bound} can also be used to recover some known inequalities for differential entropy by considering the limiting behavior as $r$ converges to one. For example, it is well known that the differential entropy of an $n$-dimensional random vector $X$ with finite second moment satisfies
\begin{align}
h(X) \le \frac{1}{2} \log \left( 2 \pi e \, \ex{\tfrac{1}{n}  \|X\|^2}  \right),  \label{eq:h_Gaussian_inq} 
\end{align}
with equality if and only if the entries of $X$ are i.i.d.\ zero-mean Gaussian. A generalization of this result in terms an arbitrary positive moment is given by
\begin{align}
h(X) & \le \log \frac{ \Gamma\left( \frac{n}{s} +1 \right)}{\Gamma\left( \frac{n}{2} +1 \right) } + \frac{n}{2} \log \pi   + \frac{n}{s} \log\left(e s \, \ex{\tfrac{1}{n}  \|X\|^{s}} \right), \label{eq:h_moment_inq}
\end{align}
for all $s > 0$.  Note that \eqref{eq:h_Gaussian_inq} corresponds to the case $s =2$. 

Inequality \eqref{eq:h_moment_inq} can be proved as an immediate consequence of  Proposition~\ref{prop:entropy_bound} and  the fact that $h_r(X)$ is non-increasing in $r$. Using properties of the beta function given in Appendix~\ref{sec:gamma_beta},  it is straightforward to verify that
\[
\lim_{r \to 1}  \psi_r(0,q) =    \left(e\,  q\right)^\frac{1}{q} \Gamma\left(\frac{1}{q} + 1\right), \quad \text{for all $q > 0$}.
\]
Combining this result with Proposition~\ref{prop:entropy_bound} and \eqref{eq:L_UB} leads to 
\begin{align*}
h(X) & \le \log \omega(S)  + \log \Gamma\left( \frac{1}{q}+ 1\right)  + \frac{1}{q} \log\left( e q \ex{\|X\|^{nq}}\right) .
\end{align*}
Using \eqref{eq:omega_UB} and making the substitution $ s= n q$ leads to \eqref{eq:h_moment_inq}.

Another example follows from the  fact that the lognormal distribution maximizes the differential entropy of a positive random variable $X$ subject to constraints on the mean and variance of $\log(X)$, and hence
\begin{align}
h(X) & \le \ex{ \log (X)}  + \frac{1}{2} \log\left( 2 \pi e \var( \log (X))\right), \label{eq:h_logmoment_inq}
\end{align}
with equality if  and only if $X$ is lognormal.  In Appendix~\ref{proof:eq:h_logmoment_inq}, it is shown how this inequality can be proved using our two-moment inequalities, by studying the behavior as both $p$ and $q$ converge to zero as $r$ increases to one.

%
%

\section{Mutual Information Bounds}\label{sec:MI_bounds}

\subsection{Relative entropy and chi-square divergence}
Let $P$ and $Q$  be distributions defined on common probability space that that have densities $p$ and $q$ with respect to a dominated measure $\mu$. The relative entropy (or Kullback--Leibler) divergence is defined according to
\[
\DKL{P}{Q} = \int p  \log \left( \frac{ p}{ q}  \right) \dd \mu,
\]
and the chi-square divergence is defined according to, 
\[
\chi^2(P,Q)  = \int \frac{\left( p -q \right)^2}{q}  \dd \mu .
\]

The chi-square divergence is equal to the squared $L_2$ distance between the densities scaled densities $p/\sqrt{q}$ and $\sqrt{q}$. The chi-square  can also be interpreted as the first non-zero term in the power series expansion of the relative entropy \cite[Lemma~4]{nielsen:2014}. More generally, the chi-square  provides an upper bound on the relative entropy, via 
\begin{align}
\DKL{P}{Q} & \le \log( 1+ \chi^2(P,Q)). \label{eq:MI_chi2_inq}
\end{align}
The proof of this inequality follows straightforwardly from Jensen's inequality and the concavity of the logarithm; see e.g.,~\cite[Theorem~5]{gibbs:2002}.

Given a random pair $(X,Y)$ the mutual information  between $X$ and $Y$ is defined according to
\[
I(X;Y)= \DKL{P_{X,Y}}{P_{X} \times P_Y}. 
\]
From \eqref{eq:MI_chi2_inq}, we see the the mutual information can always be upper bounded using 
\begin{align}
I(X;Y) \le \log( 1+ \chi^2(P_{X,Y},P_X \times P_Y)). \label{eq:chi2_bound}
\end{align}
The next section provides  bounds on the mutual information that can improve upon this inequality.

\subsection{Mutual information and variance of conditional density} 
Let $(X,Y)$ be a random pair such that the conditional distribution $Y$ given $X$ has a density $f_{Y|X}(y|x)$ with respect to Lebesgue measure on $\reals^n$. Note that the marginal density of $Y$ is given by $f_Y(y) = \ex{ f_{Y|X}(y |X)}$. To simplify notation, we will write $f(y|x)$ and $f(y)$ where the subscripts are implicit. The support set of $Y$ denoted by $\set_Y$. 
 
The measure of the dependence between $X$ and $Y$ that is used in our bounds can be understood in terms of the variance of the conditional density. For each $y$, the conditional density $f(y|X)$ evaluated with a random realization of $X$ is a random variable. The variance of  this random variable is given by
\[
\var(f(y|X)) = \ex{ \left( f(y |X) - f(y) \right)^2 }, 
\]
where we have used the fact that the marginal density $f(y)$ is the expectation of $f(y|X)$. The $s$-th moment of the variance of the conditional density is defined according to 
\[
V_s(Y |X)   = \int_{\set_Y}   \|y\|^{s} \var(f(y|X)) \, \dd y .
\]
The function $V_s(Y|X)$ is nonnegative and equal to zero if and only if $X$ and $Y$ are independent.

For $t \in (0,1]$  the function $\kappa(t)$ is defined according to 
\[
\kappa(t)  = \sup_{u \in (0,\infty)}  \frac{\log(1 + u)}{u^t} .
\]
Properties of this function are given in Appendix~\ref{sec:properties_kappa}, where it is shown that $1/(e\, t)  < \kappa(t) \le 1/t $ with equality on the right when $t =1$.

We are now ready to give the main results of this section, which are bounds on the mutual information. We begin with a general upper bound in terms of the variance of the conditional density. 

\begin{prop}\label{lem:MI_var_bound}
For any $0 <t \le 1$, the mutual information satisfies 
\[
I(X; Y)
\le  \kappa(t)  \int_{\set_Y}  \left[ f(y) \right]^{ 1- 2t }  \, \left[ \var(f(y \gmid X)) \right]^t \,  \dd y.
\]
\end{prop}
\begin{proof}
We use the following series of inequalities:
\begin{align*}
I(X;Y)  & \overset{(a)}{= }\int f(y) \, \DKL{P_{X|Y=y} }{P_X}  \, \dd y\\
& \overset{(b)}{\le}  \int f(y) \, \log\left( 1+ \chi^2(P_{X|Y=y}, P_X) \right)   \, \dd y\\
& \overset{(c)}{=}  \int f(y) \, \log\left( 1+ \frac{\var(f(y \gmid X)) }{f^2(y)}  \right)   \, \dd y\\
& \overset{(d)}{\le}  \kappa(t)   \int f(y) \left(  \frac{\var(f(y \gmid X)) }{f^2(y)}  \right)^{t}    \, \dd y,
\end{align*}
where: (a) follows from the definition of mutual information; (b) follows from \eqref{eq:MI_chi2_inq}; and  (c) follows from Bayes' rule, which allows us to write the chi-square  in terms of the variance of the conditional density:
\[
\chi^2(P_{X|Y=y}, P_X)   =  \ex{ \left( \!  \frac{f(y|X)}{f(y)} - 1 \right)^2 } = \frac{ \var(f(y|X))}{f^2(y)}.
\]
Inequality (d) follows from the non-negativity of the variance and the definition of $\kappa(t)$. 
\end{proof} 

Evaluating Proposition~\ref{lem:MI_var_bound} with $t =1$ recovers the well-known inequality  $I(X;Y) \le \chi^2(P_{X,Y}, P_X \times P_Y)$. 
The next two results follow from the cases $0 < t < 1/2$ and $t = 1/2$, respectively.

\begin{prop}\label{prop:MI_entropy_inq}
For any $0 < r < 1$, the mutual information satisfies 
\[
I(X;Y)  \le  \kappa(t)   \left(  e^{ h_r(Y)}\,   V_0(Y |X) \right)^t,
\]
where  $t = (1-r)/(2-r)$. 
\end{prop}
\iflongpaper
\begin{proof}
Starting with Proposition~\ref{lem:MI_var_bound} and applying H\"older's inequality with conjugate exponents $1/(1-t) $ and $1/t$ leads to 
\begin{align*}
I(X;Y)
&  \le \kappa(t)   \left(  \int  f^r(y) \, \dd y \right)^{1-t}  \left( \int   \var(f(y \gmid X)) \,  \dd y \right)^t\\
&  = \kappa(t)  \, e^{ t\,  h_{r}(Y) }   V^t_0(Y|X),
\end{align*}
where we have used the fact that $r = (1- 2t)/(1-t)$.
\end{proof}
\fi

\begin{prop}\label{prop:MI_two_moment_inq}
For any  $p < 1 < q$, the mutual information satisfies 
\[
I(X;Y)  \le   C(\lambda) \,     \sqrt{ \frac{  \omega(\set_Y)   V^{\lambda}_{np} (Y |X) V^{1-\lambda}_{nq} (Y |X) }{(q-p)} },
\]
where $\lambda = (q - 1)/(q-p)$ and
\[
C(\lambda) = \kappa(1/2) \sqrt{  \frac{ \pi  \lambda^{-\lambda} (1-\lambda)^{-(1-\lambda)}}{ \sin( \pi \lambda)}}.
\]
\end{prop}
\begin{proof}
Evaluating Proposition~\ref{lem:MI_var_bound} with $t  = 1/2$ gives
\[
I(X;Y) \le \kappa(1/2) \int_{\set_Y} \sqrt{  \var(f(y \gmid X))} \, \dd y .
\]
Evaluating Proposition~\ref{prop:two_moment_n} with $r  = 1/2$ leads to
\begin{align*}
\MoveEqLeft \left(  \int_{\set_Y} \sqrt{  \var(f(y \gmid X))} \, \dd y\right)^2\\
& \le   \omega(S_Y) \, \psi_{1/2}(p,q) V_{np}^{\lambda}(Y |X) V_{nq}^{1-\lambda} (Y |X).
\end{align*}
Combining these inequalities with the expression for $\psi_{1/2}(p,q)$ given in \eqref{eq:psi_one_half} completes the proof. 
\end{proof}

The contribution of Propositions~\ref{prop:MI_entropy_inq} and \ref{prop:MI_two_moment_inq} is that they provide bounds on the mutual information in terms of quantities that can be easy to characterize. One application of these bounds is to establish conditions under which the mutual information corresponding to a sequence of random pairs $(X_k,Y_k)$ converges to zero.  
%
In this case, Proposition~\ref{prop:MI_entropy_inq} provides a sufficient condition in terms of the R\'enyi entropy of $Y_n$ and the function $V_0(Y_n |X_n)$, while
%
 Proposition~\ref{prop:MI_two_moment_inq} provides a sufficient condition in terms of $V_{s}(Y_n |X_n)$ evaluated with two difference values of $s$. These conditions are summarized in the following result. 

\begin{prop}\label{prop:sufficient_conditions}
Let $(X_k,Y_k)$ be a sequence of random pairs such the conditional distribution $Y_k$ given $X_k$ has a density on $\reals^n$. The following are sufficient conditions under which the mutual information of $I(X_k; Y_k)$ converges to zero as $k$ increases to infinity:
\begin{enumerate}[(i)]
\item There exists $0 < r< 1$ such that
\begin{align}
\lim_{k \to \infty}e^{h_r(Y_k)}  V_{0}(Y_k |X_k) & = 0. \notag
\end{align}
\item There exists $p < 1< q$ such that
\begin{align}
\lim_{k \to \infty} V^{q-1}_{np}(Y_k |X_k)  V^{1-p}_{nq}(Y_k |X_k)  = 0.  \notag
\end{align}
\end{enumerate}
\end{prop}

\subsection{Properties of the bounds}\label{sec:prop_V}

The function $V_s(Y|X)$ has a number of interesting properties. The variance of the conditional density  can be expressed in terms of an expectation with respect to two independent random variables $X_1$ and $X_2$ with the same distribution as $X$ via the decomposition:
\[
\var(f(y|X)) = \ex{  f(y |X)  f(y |X)  -  f(y |X_1)  f(y |X_2) }.
\]
Consequently, by swapping the order of the integration and expectation we obtain 
\begin{equation}
V_s(Y |X)   = \ex{ K_s(X,X) - K_s(X_1,X_2)}, \label{eq:Vp_alt} 
\end{equation} 
where
\begin{equation*}
K_s(x_1,x_2) =  \int \|y\|^s f(y |x_1)  f(y |x_2) \, \dd y.
\end{equation*}
The function $K_s(x_1,x_2)$ is a positive definite kernel that does not depend on the distribution of $X$. For $s=0$, this kernel has been studied previously in the machine learning literature \cite{jebara:2004}, where  it is  referred to as the expected likelihood kernel.  


The variance of the conditional density also satisfies a data-processing inequality. Suppose that $U \to X \to Y$ forms a Markov chain. Then, the square of the conditional density of $Y$ given $U$ can be expressed as 
\[
f^2_{Y|U}(y|u) =  \ex{ f_{Y|X}(y|X'_1) f_{Y|X}(y|X'_2) \gmid U = u} ,
\]
where  $(U, X'_1,X'_2) \sim P_{U} P_{X_1 | U} P_{X_2| U}$. Combining this expression with \eqref{eq:Vp_alt} yields  
\begin{align}
V_s(Y |U)   = \ex{ K_s(X'_1,X'_2) - K_s(X_1,X_2)},  \label{eq:V_YgU} 
\end{align}
where we recall that $(X_1,X_2)$ are independent copies of $X.$ 

%

Finally, it is easy to verify that the function $V_s(Y)$ satisfies 
\begin{equation*}
V_{s}(a Y | X) = |a|^{s - n} V_s(Y|X), \quad \text{ for all $a \ne 0$}.
\end{equation*}
Using this scaling relationship we see that the sufficient conditions in Proposition~\ref{prop:sufficient_conditions} are invariant to scaling of $Y$. 

\subsection{Example with Gaussian noise}

We now provide a specific example of our bounds on the mutual information.  Let $(X,Y)$ be distributed according to
\begin{equation}
Y = X + W,  \label{eq:AWGN}
\end{equation}
were $W \sim \normal(0,1)$ is independent of $X$. In this case, it is well known that the mutual information satisfies 
\begin{equation}
I(X;Y) \le \frac{1}{2} \log(1 + \var(X)),
\end{equation}
where equality is attained is $X$ is Gaussian. This inequality follows straightforwardly from the fact that the Gaussian distribution maximizes differential entropy subject to a second moment constraint. 
One of the limitations of this bound is that it can be loose when the second moment is dominated by events that have small probability. In fact, it is easy to construct examples for which $X$ does not have a finite second moment and yet $I(X;Y)$ is arbitrarily close to zero.


Our results provide bounds on  $I(X;Y)$ that are significantly less sensitive to the effects of rare events. To begin, observe that the product of the conditional densities can be factored according to
\begin{align*}
 f(y | x_1) f(y|x_2)
  & = \phi\left(  \sqrt{2} \, y - \frac{x_1 + x_2}{\sqrt{2}}  \right)  \phi\left( \frac{x_1 - x_2}{\sqrt{2}}   \right),
\end{align*}
where $\phi(x) = (2\pi)^{-1/2} \exp(- x^2/2)$ is the density of the standard Gaussian distribution. Integrating with respect to $y$ leads to 
\begin{align*}
K_s(x_1,x_2)   & =
 2^{-\frac{1+s}{2}}  \ex{ \left|W + \frac{x_1 + x_2 }{\sqrt{2}} \right|^s}\phi\left( \frac{x_1 - x_2}{\sqrt{2}}   \right).
\end{align*}
For the case $s=0$, the function $K_0(x_1,x_2)$ is proportional the standard Gaussian kernel and we have
\[
V_0(Y|X)  = \frac{1}{2 \sqrt{ \pi}} \left[ 1 - \ex{ e^{ - \frac{1}{4}  \left( X_1 - X_2   \right)^2  }} \right].
\]
This expression shows that $V_0(Y|X)$ is a measure of the variation in $X$. 

A useful property of $V_{0}(Y|X)$ is that the conditions under which  it converges to zero are weaker than the conditions needed for other measures of variation, such that variance. To see why, observe that the expectation is bounded uniformly with respect to $(X_1,X_2)$. In particular, for every $\eps > 0$ and $x \in \reals$, we have
\begin{align*}
 1 - \ex{ e^{- \frac{1}{4}  \left( X_1 - X_2   \right)^2  }}  
 & \le \eps^2 + 2 \pr{|X -x| \ge \eps},
\end{align*}
where we have used the inequality $1 - e^{-x} \le x$ and the fact that $\pr{ |X_1 - X_2| \ge 2 \eps} \le 2\pr{ |X - x| \ge \eps}$. Therefore, $V_0(Y|X)$ is small provided that $X$ is close a constant value with high probability.

To study some further properties of these bounds, we now focus on the case where $X$ is a Gaussian scalar mixture generated according to
\begin{equation}
X = A \sqrt{U}  , \quad A \sim \normal(0,1), \quad U \ge 0, \label{eq:scale_mixture}
\end{equation}
with $A$ and $U$ independent.  
In this case, the expectations with respect to the kernel $K_s(x_1,x_2)$ can be computed explicitly, leading to 
\begin{align*}
V_s(Y|X)  =\cramped{ \frac{\Gamma( \frac{1 + s}{2})}{2 \pi}  }\ex{ \cramped{\left(1  + 2 U\right)^\frac{s}{2}} -  \frac{ (1 +U_1)^\frac{s}{2} ( 1+ U_2)^\frac{s}{2}  }{ \cramped{ (1 + \frac{1}{2}  (U_1 + U_2) )^\frac{s+1}{2} }}  }.
\end{align*}
It can be shown that this expression depends primarily on the magnitude of $U$. This is not surprising given that $X$ converges to a constant if and only if $U$ converges to zero. 


Our results can also be used to bound the mutual information $I(U;Y)$ by noting that $U \to X \to Y$ forms a Markov chain, and taking advantage of the characterization provided in  \eqref{eq:V_YgU}.  Letting $X_1' = A_1 \sqrt{U}$ and $X_2' = A_2 \sqrt{U}$ with $(A_1,A_2, U)$ mutually independent, leads to
 \begin{align*}
V_s(Y|U)  =\cramped{ \frac{\Gamma( \frac{1 + s}{2})}{2 \pi} } \ex{ \cramped{\left(1  + U\right)^\frac{s-1}{2}} -  \frac{ (1 +U_1)^\frac{s}{2} ( 1+ U_2)^\frac{s}{2}  }{ \cramped{ (1 + \frac{1}{2}  (U_1 + U_2) )^\frac{s+1}{2} }}   }.
\end{align*}
In this case, $V_s(Y|U)$ is a measure of the variation in $U$. To study it behavior, we  consider the simple upper bound
\begin{align*}
V_s(Y|U)   \le  \frac{\Gamma( \frac{1 + s}{2})}{2 \pi}   \pr{ U_1 \ne U_2} \ex{ \left(1  +U \right)^\frac{s-1}{2} }. 
\end{align*}
This bound shows that if $s \le 1$ then $V_s(Y|U)$ is bounded uniformly  with respect to distributions on $U$, and if $s >1$ then $V_s(Y|U)$  is bounded in terms of the $(\frac{s-1}{2})$-th moment of $U$.

In conjunction with Propositions~\ref{prop:MI_entropy_inq} and \ref{prop:MI_two_moment_inq} the function $V_{s}(Y|U)$ provide bounds on the mutual information $I(U;Y)$ that can be expressed in terms of simple expectations involving two independent copies of $U$. Figure~\ref{fig:MI_bound} provides an illustration of the upper bound in  Proposition~\ref{prop:MI_two_moment_inq} for the case where  $U$ is a discrete random variable supported on two-points and $X$ and $Y$  are generated according to \eqref{eq:AWGN} and \eqref{eq:scale_mixture}.  This example shows that there exist sequences of distributions for which our upper bounds on the mutual information converges to zero while the chi-square divergence between $P_{XY}$ and $P_X \times P_Y$ is bounded away from zero. 

\begin{figure}
\centering
\begin{tikzpicture}
    	\node[anchor=south west,inner sep=0] at (0,0) {\includegraphics[height=2in]{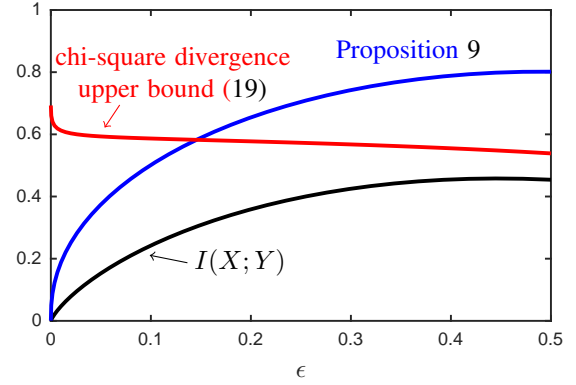}};
        \node[ align =center, text = {rgb:red,0;green,0;blue,1}] at (5.3,4.4) { Proposition~\ref{prop:MI_two_moment_inq} };
                \node[ align =center, text = {rgb:red,1;green,0;blue,0}] at (2.1,4.1) {chi-square divergence\\ upper bound \eqref{eq:chi2_bound} };
                   \draw [<-,color = {rgb:red,1;green,0;blue,0}] (1.2, 3.4) -- (1.4,3.7);
              \node[ align =center] at (3,1.6) { $I(X;Y)$};
       \draw [<-] (1.8, 1.7) -- (2.3,1.6);
  
\end{tikzpicture}
\vspace*{-.2cm}
\caption{\label{fig:MI_bound} Bounds  on the mutual information $I(U;Y)$ when $U \sim (1-\eps) \delta_1  + \eps \delta_{a(\eps)}$, with $a(\eps) = 1 + 1/\sqrt{\eps}$, and $X$ and $Y$  are generated according to \eqref{eq:AWGN} and \eqref{eq:scale_mixture}. The bound from Proposition~\ref{prop:MI_two_moment_inq} is evaluated with $p=0$ and $q = 2$. }
\label{default}
\vspace*{-.4cm}
\end{figure}

\section{Conclusion} 

This paper provides bounds on R\'enyi entropy and mutual information that are based on a relatively simple two-moment inequality. One of the main takeaways from our analysis is that sometimes two carefully chosen moments are all that is needed to provide an accurate characterization. Extensions to inequalities with more moments are also worth exploring.

\appendices

\section{The Gamma and Beta Functions}\label{sec:gamma_beta}

This section reviews some properties of the gamma and beta functions. For  $x > 0$, the gamma function is defined according to 
$ \Gamma(x) = \int_0^\infty t^{x-1} e^{-t} \, \dd t.$ 
Binet's formula the logarithm of the gamma function \cite[Theorem 1.6.3]{andrews:1999} gives
\begin{align}
\log \Gamma(x) = \left( x - \frac{1}{2} \right) \log x  - x + \frac{1}{2} \log(2 \pi)+ \theta(x),  \label{eq:binet} 
\end{align}
where the remainder term $\theta(x)$ is convex and non-increasing with $\lim_{x \to 0} \theta(x) = \infty$ and $\lim_{x \to \infty} \theta(x) = 0$. Euler's reflection formula \cite[Theorem 1.2.1]{andrews:1999} gives
\begin{align}
\Gamma(x) \Gamma(1-x) = \frac{ \pi}{ \sin (\pi x) } . \label{eq:euler} 
\end{align}

For $x,y > 0$ the beta function is defined according to $\Beta(x,y) = \Gamma(x) \Gamma(y) / \Gamma(x+y)$. 
The beta function  can also be expressed in integral form as \cite[pg.~7]{andrews:1999}
\begin{align}
\Beta(x,y) 
=  \int_0^\infty \frac{s^{x-1}}{(1+s)^{x+y}} \, \dd s. \label{eq:beta_alt}
\end{align}
Recall that  $\widetilde{\Beta}(x,y) = \Beta(x,y) (x+y)^{x+y} x^{-x} y^{-y}$. Using \eqref{eq:binet} leads to 
\begin{align}
 \log\left(  \widetilde{\Beta}(x,y)  \sqrt{ \frac{x\, y}{2 \pi( x\!+\!y) } } \right) =  \theta(x)  +\theta(y)  - \theta(x+y). \label{eq:Beta_tilde_alt}
\end{align}
It can also be shown that \cite[Equation~(2) pg.~2]{grenie:2015}
\begin{align}
  \widetilde{\Beta}(x,y)  \ge \frac{x + y}{x y}
 . \label{eq:Beta_tilde_alt2}
\end{align}

%
%

\section{Details for R\'enyi Entropy Examples} 

This appendix studies properties of the two-moment inequalities for R\'enyi entropy described in Section~\ref{eq:entropy_bounds}.  

\subsection{Lognormal distribution} \label{sec:lognormal_proof}

Let $X$ be a lognormal random variable with parameters $(\mu,\sigma^2)$ and consider the parametrization
\begin{align*}
p &= \frac{1-r}{r} - (1-\lambda)  \sqrt{ \frac{(1-r) \, u }{r\lambda(1-\lambda) }   }  \\
q &= \frac{1-r}{r}  +\lambda  \,  \sqrt{  \frac{(1-r) \,u  }{r\lambda (1-\lambda)}   }.
\end{align*}
where $\lambda \in(0,1)$ and $u \in (0,\infty)$. Then, we have
\begin{align*}
 \psi_r(p,q)  & = \widetilde{\Beta}\left(\frac{r \lambda}{1-r} , \frac{r (1-\lambda)}{1-r} \right)   \sqrt{  \frac{r \lambda (1-\lambda)  }{(1-r)  \, u  }} \\
L_r(X;p,q)  & =  \mu  + \frac{1}{2} \left( \frac{1-r}{ r} \right) \sigma^2 + \frac{1}{2} u \sigma^2.
\end{align*}
Combining these expressions with \eqref{eq:Beta_tilde_alt} leads to 
\begin{align}
\Delta_r(X;p,q) 
&=  \theta\Big( \frac{r \lambda}{1\!-\!r} \Big) + \theta\Big( \frac{r (1\!-\!\lambda)}{1-r} \Big) - \theta\Big( \frac{r}{1\!-\!r} \Big) \notag \\
& \quad + \frac{1}{2} u \sigma^2    -  \frac{1}{2} \log\left( u  \sigma^2\right) - \frac{1}{2} \log( r^\frac{1}{r-1})\label{eq:Delta_lognormal_c}.
\end{align}

We now characterize the minimum with respect to the parameters $(\lambda, u)$. Note that the mapping $\lambda \mapsto \theta( \frac{r \lambda}{1-r}) + \theta( \frac{r (1-\lambda)}{1-r})$ is convex and symmetric about the point $\lambda = 1/2$. Therefore, the minimum with respect to $\lambda$ is attained at $\lambda = 1/2$. Meanwhile,  mapping $u \mapsto u \sigma^2 - \log(u \sigma^2)$ is convex and attains it minimum at $u = 1/\sigma^2$. Evaluating \eqref{eq:Delta_lognormal_c} with these values, we see that the optimal two-moment inequality can be expressed as
\begin{align*}
\Delta_r(X) & = 2 \theta\left( \frac{r}{2(1-r)} \right)  - \theta\left( \frac{r}{1-r} \right) + \frac{1}{2} \log\left( e\,  r^\frac{1}{1-r} \right). 
\end{align*}
By \eqref{eq:Beta_tilde_alt}, this expression is equivalent to \eqref{eq:binet}. Moreover, the fact that $\Delta_r(X)$ decreases to zero as $r$ increases to one follows from the fact that $\theta(x)$ decreases to zero and $x$ increases to infinity.  

Next, we express gap in terms of the pair $(p,q)$.  Comparing the difference between $\Delta_r(X;p,q)$ and $\Delta_r(X)$ leads to
\begin{align*}
\Delta_r(X; p,q) & =   \Delta_r(X) + \frac{1}{2} \varphi\left( \frac{r \lambda (1-\lambda) }{1-r} (q - p)^2  \sigma^2 \right) \\
& \quad + \theta\Big( \frac{r \lambda}{1\!-\!r} \Big) + \theta\Big( \frac{r (1\!-\!\lambda)}{1-r} \Big)  - 2 \theta\Big( \frac{r}{2(1\!-\!r)} \Big),
\end{align*}
where $\varphi(x) = x  - \log(x) - 1$. In particular, if $p =0$, then we obtain the simplified expression
\begin{align*}
 \Delta_r(X; 0,q) & =  \Delta_r(X)   + \frac{1}{2} \varphi\left( \Big( q -\frac{1-r}{r}\Big) \sigma^2  \right)  \\
& \quad +  \theta\Big( \frac{r}{1-r} - \frac{1}{q} \Big)  + \theta\Big(\frac{1}{q} \Big)  - 2 \theta\Big( \frac{r}{2(1\!-\!r)}\Big).
\end{align*}
This characterization shows that the gap of the optimal one-moment inequality $\widetilde{\Delta}_r(X)$ increases to infinity in the limit as either $\sigma^2 \to 0$ or $\sigma^2 \to \infty$.

\subsection{Multivariate Gaussian distribution}

Let  $Y \sim \normal(0,I_n)$ is an $n$-dimensional Gaussian vector and consider the parametrization
\begin{align*}
p &= \frac{1-r}{r} -  \frac{1-\lambda}{r}  \sqrt{ \frac{ 2(1-r) \, z}{\lambda (1-\lambda) \, n}    }  \\
q &= \frac{1-r}{r}  + \frac{\lambda}{r} \sqrt{ \frac{ 2(1-r) \, z}{\lambda (1-\lambda) \, n}    } .
\end{align*}
where $\lambda \in(0,1)$ and $z \in (0,\infty)$. The, we have
\begin{align*}
\log \omega(\set_Y) & = \frac{n}{2} \log \pi -  \log\left( \frac{n}{2}\right)   - \log \Gamma\left( \frac{n}{2} \right) \\
 \psi_r(p,q)  & = \widetilde{\Beta}\left(\frac{r \lambda}{1-r} , \frac{r (1-\lambda)}{1-r} \right)   \sqrt{  \frac{r \lambda (1-\lambda)  }{(1-r)  }}  \sqrt{ \frac{n r}{2z} }. 
\end{align*}
Furthermore, if 
\begin{align}
(1-\lambda) \sqrt{ \frac{ 2 (1-r) z}{ \lambda (1-\lambda) n}} <1, \label{eq:lambda_z_cond}
\end{align}
then $ L_r(\|Y\|^n ;p,q)  $ is finite and is given by 
\begin{align*}
 L_r(\|Y\|^n ;p,q)  
  & = Q_{r,n}(\lambda,z) +  \frac{n}{2} \log 2  \\
  & \quad  + \frac{r}{1-r} \left[ \log \Gamma\left( \frac{n}{2r}\right) - \log \Gamma\left( \frac{n}{2}\right) \right],
\end{align*}
where
\begin{align}
 Q_{r,n}(\lambda,z )
&  = \frac{r \lambda }{1-r}   \log \Gamma\bigg( \frac{n}{2 r} -  \frac{1-\lambda}{r} \sqrt{ \frac{(1-r) n z}{2 \lambda(1-\lambda)}  } \, \bigg)  \notag\\
&  \quad + \frac{r (1- \lambda)}{1-r}   \log \Gamma\bigg(  \frac{n}{2 r}  +  \frac{\lambda}{r} \sqrt{ \frac{(1-r)n z}{2 \lambda(1-\lambda)}  } \, \bigg) \notag \\
& \quad  -\frac{r}{1-r}  \log \Gamma\left( \frac{n}{2r} \right)  . \label{eq:Q}
\end{align}
Combining these expressions and then using \eqref{eq:binet} and \eqref{eq:Beta_tilde_alt} leads to 
\begin{align}
\Delta_r(Y;p,q) & =  \theta\Big( \frac{r \lambda}{1\!-\!r} \Big) + \theta\Big( \frac{r (1\!-\!\lambda)}{1-r} \Big) - \theta\Big( \frac{r}{1\!-\!r} \Big) \notag \\
 & \quad + Q_{r,n}(z,\lambda) - \frac{1}{2} \log z  - \frac{1}{2} \log\left( r^\frac{1}{r-1} \right) \notag \\
& \quad +\frac{r}{1-r} \theta\left( \frac{n}{2r} \right) - \frac{1}{1-r} \theta\left( \frac{n}{2} \right). \label{eq:Delta_Gaussian_c}  
\end{align}

Next, we study some properties of  $Q_{r,n}(\lambda,z)$. The decommissions \eqref{eq:binet} shows that the logarithm of the gamma function can expressed as the sum of convex functions:
\begin{align*}
\log \Gamma(x) = \varphi(x)  + \frac{1}{2} \log\left( \frac{1}{x} \right)  + \frac{1}{2} \log ( 2\pi) - 1 + \theta(x),
\end{align*}
where $\varphi(x) = x \log x   + 1  - x$. Starting with the definition of $Q(\lambda, z)$ and then using Jensen's inequality yields 
\begin{align*}
&  Q_{r,n}(z,\lambda)  \\
& \ge   \frac{r \lambda }{1-r}  \varphi\bigg( \frac{n}{2 r} -  \frac{1-\lambda}{r} \sqrt{ \frac{(1-r) n z}{2 \lambda(1-\lambda)}  } \, \bigg)  \notag \\
&  \quad + \frac{r (1- \lambda)}{1-r}  \varphi\bigg(  \frac{n}{2 r}  +  \frac{\lambda}{r} \sqrt{ \frac{(1-r)n z}{2 \lambda(1-\lambda)}  } \, \bigg)  -\frac{r}{1-r} \varphi\left( \frac{n}{2r} \right)  \\
&  =    \frac{  \lambda}{a}  \varphi\left( 1 - \sqrt{\left( \tfrac{1- \lambda}{\lambda} \right)  a z } \right)   +  \frac{(1\!-\!\lambda)}{a} \varphi\bigg( 1 + \sqrt{\left( \tfrac{\lambda}{1- \lambda} \right)  a z } \bigg) ,
\end{align*}
where $a = 2 (1-r) /n$. Using the inequality $\varphi(x) \ge  (3/2)(x-1)^2/(x+2)$ leads to
\begin{align}
Q_{r,n}(\lambda,z) 
& \ge \frac{z}{2} \left[ \left(1 - \sqrt{\left( \tfrac{1- \lambda}{\lambda} \right)  b z } \right)\bigg(1 + \sqrt{\left( \tfrac{\lambda}{1- \lambda} \right)  b z } \bigg) \right]^{-1} \notag \\
& \ge \frac{z}{2}\bigg(1 + \sqrt{\left( \tfrac{\lambda}{1- \lambda} \right)  b \,  z } \bigg)^{-1}  ,  \label{eq:Q_LB}
\end{align}
where $b=  2 (1-r) /( 9n)$.

Observe that the right-hand side of \eqref{eq:Q_LB} converges to $z/2$ as $n$ increases to infinity. It turns out this limiting behavior is tight. Using \eqref{eq:binet}, it is straightforward to show that $Q_{n}(\lambda,z)$ converges pointwise to $z/2$ as $n$ increases to infinity,  that is
\begin{align}
\lim_{n \to \infty}  Q_{r,n}(\lambda,z)  = \frac{1}{2} z. \label{eq:Q_lim}
\end{align}
for any fixed pair $(\lambda, z) \in (0,1) \times (0,\infty)$.

\subsection{Proof of Proposition~\ref{prop:Delta_Gaussian_proof}} \label{sec:Delta_Gaussian_proof}
Let $D = (0,1) \times (0,\infty)$.  For fixed $r \in (0,1)$ we use $Q_{n}(\lambda, z)$ to denote the function $Q_{r,n}(\lambda, z)$ defined in \eqref{eq:Q} and we use $G_n(\lambda, z)$ to denote the right-hand side of \eqref{eq:Delta_Gaussian_c}. These functions are defined to be equal to positive infinity for any pair $(\lambda, z) \in D$ such that \eqref{eq:lambda_z_cond} does not hold. 

Note that the terms $\theta(n/(2r))$ and $\theta(n/2)$ converge to zero in the limit as $n$ increases to infinity. In conjunction with \eqref{eq:Q_lim}, this shows that  $G_n(\lambda, z) $ converges pointwise to a limit $G(\lambda, z) $  given by 
\begin{align*}
G(\lambda, z) & = \theta\Big( \frac{r \lambda}{1\!-\!r} \Big) + \theta\Big( \frac{r (1\!-\!\lambda)}{1-r} \Big) - \theta\Big( \frac{r}{1\!-\!r} \Big) \notag \\
& \quad + \frac{1}{2}z    -  \frac{1}{2} \log\left( z \right) - \frac{1}{2} \log( r^\frac{1}{r-1}).
\end{align*}
At this point, the  correspondence with the lognormal distribution can be seen from the fact that $G(\lambda, z) $ is equal to the right-hand side of \eqref{eq:Delta_lognormal_c} evaluated with $u \sigma^2 = z$. 

To show that the gap corresponding to the  lognormal distribution provides an upper bound on the limit, we use
\begin{align}
\limsup_{n \to \infty} \Delta_r(Y) & = \limsup_{n \to \infty}  \inf_{(\lambda, z) \in D} G_n(\lambda, z)  \notag \\
& \le  \inf_{(\lambda, z) \in D}   \limsup_{n \to \infty} G_n(\lambda, z) \notag \\
&=\inf_{(\lambda, z) \in D} G(\lambda, z) \notag \\
& = \Delta_r(X).  \label{eq:G_n_a} 
\end{align}
Here, the last equality follows from the analysis in Appendix~\ref{sec:lognormal_proof}, which shows that the minimum of $G(\lambda, z)$ is a attained at $\lambda = 1/2$ and $z = 1$. 

To prove the lower bound requires a bit more work. Fix any $\eps \in (0,1)$ and let $D_\eps = (0,1-\eps] \times (0,\infty)$. Using the lower bound on $Q_{n}(\lambda, z)$ given in  \eqref{eq:Q_LB}, it can be verified that
\begin{align*}
\liminf_{ n \to \infty} \inf_{(\lambda,z) \in D_\eps } \left[  Q_{r,n} (z,\lambda)   - \frac{1}{2} \log z \right] \ge \frac{1}{2} .
\end{align*}
Consequently, we have
\begin{align}
\liminf_{n \to \infty}  \inf_{(\lambda, z) \in D_\eps }  G_n(\lambda, z) 
& =   \inf_{(\lambda, z) \in D_\eps}  G(\lambda, z)  \ge \Delta_r(X).   \label{eq:G_n_b} 
\end{align}
To complete the proof we will show that for any sequence  $\lambda_n$ that converges to one as $n$ increases to infinity, we have
\begin{align}
\liminf_{n \to \infty}  \inf_{z \in (0,\infty) }   G_n(\lambda_n, z) = \infty.\label{eq:G_n_c} 
\end{align}
To see why this is the case, note that by \eqref{eq:Beta_tilde_alt} and \eqref{eq:Beta_tilde_alt2},
\begin{align*}
 \theta\big( \tfrac{r \lambda}{1-r} \big) + \theta\big( \tfrac{r (1-\lambda)}{1-r} \big) - \theta\big( \tfrac{r}{1-r} \big) \ge  \frac{1}{2} \log\Big( \frac{1-r}{2 \pi r \lambda   (1\!-\!\lambda) } \Big). 
\end{align*}
Therefore, we can write 
\begin{align}
G_n\left(\lambda,z \right) \ge Q_{n}(\lambda,z) - \frac{1}{2}\log\left( \lambda (1-\lambda) z\right) + c_{n}, \label{eq:G_LB}
\end{align}
where $c_{n}$ is bounded uniformly for all $n$. Making the substitution $u = \lambda  (1-\lambda) z $, we obtain
\begin{align*}
\inf_{z> 0}  G_n\left(\lambda, z \right) & \ge \inf_{u > 0}  \left[ Q_{n}\left( \lambda, \frac{u}{\lambda(1-\lambda)}\right) - \frac{1}{2} \log u \right] + c_{n}.
\end{align*}
Next, let  $b_n =  2 (1-r)/(9n)$. The lower bound in \eqref{eq:Q_LB}  leads to
\begin{align}
\MoveEqLeft[1]  \inf_{u>0} \left[ Q_{n}\left(\lambda, \frac{u}{\lambda (1-\lambda)} \right) - \frac{1}{2} \log u \right] \notag \\ 
& \ge \inf_{u >0}  \left[   \frac{u}{2 \lambda   } \left( \frac{ 1}{ 1 - \lambda  +\sqrt{  b_n  u } } \right)   - \frac{1}{2} \log u \right].  \label{eq:Q_LB3} 
\end{align}
The limiting behavior in \eqref{eq:G_n_c} can now be seen as a consequence of  \eqref{eq:G_LB} and the fact that, for any sequence $\lambda_n$ converging to one,  the right-hand side \eqref{eq:Q_LB3} increases without bound as $n$ increases. Combining \eqref{eq:G_n_a}, \eqref{eq:G_n_b}, and \eqref{eq:G_n_c} establishes that the large $n$ limit of $\Delta_r(Y)$ exists and is equal to $\Delta_r(X)$. This concludes the proof of Proposition~\ref{prop:Delta_Gaussian_proof}



\subsection{Proof of Inequality~\eqref{eq:h_logmoment_inq}} \label{proof:eq:h_logmoment_inq} 

Given any $\lambda \in (0,1)$ and $u \in (0,\infty)$ let
\begin{align*}
p(r) = \frac{1-r}{r} - \sqrt{ \frac{ 1-r}{r} \left(\frac{ 1-\lambda}{\lambda}\right) u }  \\
q(r) = \frac{1-r}{r}  +  \sqrt{ \frac{ 1-r}{r} \left( \frac{\lambda}{1-\lambda}\right) u } .
\end{align*}
We need the following results, which characterize the terms in Proposition~\ref{prop:entropy_bound}  in the limit as $r$ increases to one.  

\begin{lemma}\label{lem:psi_big_r} The function  $\psi_r(p(r),q(r))$ satisfies
\[
\lim_{r \to 1}  \psi_r(p(r),q(r)) =  \sqrt{ \frac{2 \pi}{u} }
\]
\end{lemma}
\begin{proof}
Starting with  \eqref{eq:Beta_tilde_alt}, we can write
\begin{align}
\MoveEqLeft \psi_r(p,q)  = \frac{1}{q-p}  \sqrt{ \frac{ 2 \pi  (1-r) }{r  \lambda(1-\lambda}}  \notag \\
& \quad \times  \exp\left(  \theta\Big( \frac{r \lambda}{1\!-\!r} \Big) + \theta\Big( \frac{r (1\!-\!\lambda)}{1-r} \Big) - \theta\Big( \frac{r}{1\!-\!r} \Big)   \right).  \notag
\end{align}
As $r$ converges to one the terms in the exponent converge to zero. Noting that $q(r) - p(r)  = \sqrt{ r \lambda(1-\lambda)/(1-r)}$ completes the proof.
\end{proof}

\begin{lemma}\label{lem:L_big_r}
If  $X$ is a random variable such that $s \mapsto \ex{|X\|^s}$ is finite in a neighborhood of zero, then $\ex{ \log(X)}$ and $\var( \log(X))$ are finite, and 
\[
\lim_{r \to1} L_r(X; p(r),q(r)) =   \ex{ \log |X|} +  \frac{u}{2}  \var(\log |X|) .
\]
\end{lemma}
\begin{proof}
Let $\Lambda(s) = \log( \ex{ |X|^s})$. The assumption that $\ex{|X|^s}$ is finite in a neighborhood of zero means that  $\ex{ (\log |X|)^m}$ is finite for all positive integers $m$ that $\Lambda(s)$ is real-analytic in a neighborhood of zero, that is there exists constants $\delta > 0 $ and $C < \infty$, depending on $X$, such that
\[
\left| \Lambda(s) -  a s  +  b s^2  \right| \le C\,  |s|^3, \quad \text{for all $|s| \le \delta$}, 
\]
where $a = \ex{ \log |X|}$ and $b = \frac{1}{2} \var(|X|)$. Consequently, for all $r$ such that $1-\delta < p(r) < (1-r)/r < q(r) < 1 + \delta$, it follows that
\begin{align*}
\MoveEqLeft \left| L_r(X ; p(r),q(r))  - a  - \left ( \tfrac{1-r}{r}  + u\right )  b \right|\\
& \quad  \le C \frac{r}{1-r}  \left( \lambda |p(r)|^3 + (1-\lambda) |q(r)|^3\right) .
\end{align*}
Taking the limit as $r$ increases to one completes the proof. 
\end{proof}

We are now ready to prove Inequality~\eqref{eq:h_logmoment_inq}.  Combining Proposition~\ref{prop:entropy_bound} with Lemma~\ref{lem:psi_big_r} and Lemma~\ref{lem:L_big_r} yields
\begin{align*}
\limsup_{r \to \infty} h_r(X) &\le  \frac{1}{2} \log\left( \frac{2 \pi}{u}  \right)  + \ex{ \log X} + \frac{u}{2}  \var( \log X).
\end{align*}
The stated inequality follows from evaluating the right-hand side with $u = 1/\var(\log X)$ and recalling that $h(X)$ corresponds to the limit of $h_r(X)$ as $r$ increases to one.

\section{Properties of Logarithm-Power Ratio} \label{sec:properties_kappa}

This section studies properties of the function $\kappa: (0,1] \to \reals_+$ defined by
\begin{align}
\kappa(t)  = \sup_{u \in (0,\infty)} \frac{ \log(1 + u)}{u^t}  .
\end{align} 
For $t = 1$, the bound $\log(1+u) \le u$ means that $\kappa(1) \le 1$. Noting $\lim_{u \to 0} \log(1+u)/u = 1$ shows that this inequality is tight, and thus $\kappa(1) = 1$.  For any $t \in (0,1)$,  it can be verified via differentiation that the supremum is attained on $(0,\infty)$ by the unique solution $u^*_t$ to the fixed-point equation
\begin{align}
u =  t   (1+u) \log(1 +u). \label{eq:u_fp}
\end{align}
The solution to this equation can be expressed as
\[
u_t^* =  \exp\left( W\left( -\tfrac{1}{t}\exp\left( -\tfrac{1}{t}\right)\right) +  \tfrac{1}{t} \right) - 1,
\]
where Lambert's function $W(z)$ is the  solution to the equation $z = x \exp(x)$ on the interval on $[-1, \infty)$.

\begin{lemma}
The function $g(t) = t \kappa(t)$ is nondecreasing on $(0,1]$ with $\lim_{t \to 0} g(t) = 1/e$ and $g(1) = 1$. 
\end{lemma}
\begin{proof}
The fact that $g(1) = 1$ follows from $\kappa(1)=1$. 
By the envelope theorem \cite{milgrom:2002}, the derivative of $g(t)$ can be expressed as
\[
g'(t)  = \left( \frac{1}{t} - \log(u^*_t )  \right)  g(t) .
\]
Therefore, the derivative satisfies
\begin{align*}
g'(t) \ge 0 & \iff  \frac{1}{t}  - \log(u^*_t) \ge  0  \\
&  \iff \frac{(1+u_t^*) \log(1+u_t^*)}{u^*_t} - \log(u_t^*)  > 0 \\
&  \iff (1+u_t^*) \log(1+u_t^*) \ge  u_t^* \log(u_t^*) .
\end{align*}
Noting that $u \mapsto u \log u$ is negative on $(0,1)$ and nonnegative and nondecreasing on $[1, \infty)$ shows that the last condition is always satisfied, and hence $g'(t)$ is nonnegative. 

To prove the small $t$ limit we can rearrange  \eqref{eq:u_fp} to see that $u^*_t$ satisfies 
\begin{equation}
 \frac{u^*_t}{(1+u^*_t) \log(1+u^*_t)} = t,  \label{eq:u_fp2} 
\end{equation}
and hence
\begin{equation}
\log(g(t))  = \log\left( \frac{u^*_t}{1+u^*_t} \right)  - \frac{u^*_t \log u^*_t}{(1+u^*_t) \log(1+u^*_t)} . \label{eq:gt_alt} 
\end{equation}
Now, as $t$ decreases to zero,  \eqref{eq:u_fp2} shows that  $u_t^*$ increases to infinity. By \eqref{eq:gt_alt}, it then follows that $\log(g(t))$ converges to negative one, which proves the desired limit. 
\end{proof}

%
%
%
%
%
\bibliographystyle{IEEEtran}

\bibliography{two_moment_arxiv_v2.bbl}

\end{document}